\documentclass[a4paper,10pt]{amsart}
\usepackage{graphicx}
\usepackage{amssymb,amsmath,amstext,amscd}
\theoremstyle{plain}
\newtheorem{thm}{Theorem}[section]
\newtheorem{proposition}{Proposition}[section]

\newtheorem*{main theorem}{Theorem}

\theoremstyle{definition}

\newtheorem*{definition}{Definition}

\begin{document}

\title{On continuous causal isomorphisms}
\author{Do-Hyung Kim}
\address{Department of Mathematics, College of Natural Science, Dankook University,
San 29, Anseo-dong, Dongnam-gu, Cheonan-si, Chungnam, 330-714,
Republic of Korea} \email{mathph@dankook.ac.kr}


\begin{abstract}
It is shown that continuous causal isomorphisms on two-dimensional
Minkowski spacetime can be characterized by the invariance of wave
equations.
\end{abstract}

\maketitle

\section{Introduction} \label{section:1}

In 1964, Zeeman has shown that general form of causal isomorphism
$F$ defined on $\mathbb{R}^n_1$, when $n \geq 3$, has the form
$F(\mathbf{x}) = a A \mathbf{x} + \mathbf{b}$, where $a$ is a
positive real number and $A$ is an orthochronous
matrix.(\cite{Zeeman}) As Zeeman remarked, his result does not
hold when $n=2$. In 2010, the general form of causal isomorphism
defined on $\mathbb{R}^2_1$ was clarified in \cite{CQG3} and
\cite{CQG4}. When these results were obtained, the results
strongly suggested that there are some kind of close relationship
between causal isomorphism and wave equation and thus, in
\cite{Low}, Low officially proposed the question how to
characterize causal isomorphisms on $\mathbb{R}^2_1$ in terms of
wave equations. In \cite{Wave}, it is shown that, when $n \geq 3$,
causal isomorphisms on $\mathbb{R}^n_1$ can be characterized by
invariance of wave equations and in \cite{Kim}, even if $n=2$,
causal isomorphisms can be characterized by the invariance of wave
equations. However, one of characteristic differences between
causal isomorphisms on $\mathbb{R}^n_1$ ($n \geq 3$) and
$\mathbb{R}^2_1$ is that causal isomorphisms on $\mathbb{R}^2_1$
are not necessarily differentiable whereas they are $C^\infty$
when $n \geq 3$. Since causal isomorphisms are necessarily smooth
when $n \geq 3$, in \cite{Wave}, we don't need to worry about the
smoothness of causal isomorphisms. In contrast, in \cite{Kim}, it
must be explicitly assumed that the causal isomorphisms are at
least $C^2$ to ensure that we can apply wave operators since there
are non-differentiable causal isomorphisms. In other words, it
only remains that how we can characterize $C^0$ and $C^1$ causal
isomorphisms on $\mathbb{R}^2_1$ in terms of wave equations.

In this paper, it is shown that when $n=2$, even continuous causal
isomorphisms can be characterized by the invariance of wave
equations. In order to take derivatives of $C^0$ functions, we
need to generalize the derivatives and to this end, it is
necessary to introduce generalized functions and generalized
derivatives.

\section{An Overview on distributions} \label{section:2}

 To take derivatives of functions which are not differentiable
 in the classical sense, we need to generalize the notion of
 functions and derivatives. In this section, we introduce
 distributions(or, generalized functions) and their derivatives,
 and briefly review their basics.

Ordinarily, a real-valued function $f$ is given by specifying its
value $f(x)$ to each point $x$ in the domain of definition.
However, $f$ can be defined in another equivalent ways. For
example, let $f$ be a continuous function on $\mathbb{R}^2$ and
assume that we know the value of $\int_{\mathbb{R}^2} f(x,y)
\phi(x,y) dx dy$ for each $C^\infty$ function $\phi$ whose support
is compact. Then, by use of convolution, we can obtain the value
$f(x,y)$ for each $(x,y) \in \mathbb{R}^2$.(For details, see
Theorem 7.7 in \cite{Folland}). Therefore, the value of
$\int_{\mathbb{R}^2} f(x,y) \phi(x,y) dx dy$ for each $C^\infty$
function $\phi$ whose support is compact, can be used to define a
function and this gives us a way to generalize the concept of
ordinary functions. Hence, we proceed as follows.

\begin{definition}
 By test functions on $\mathbb{R}^2$, we mean $C^\infty$ functions
 defined on $\mathbb{R}^2$
with compact supports and we denote the set of all test functions
by $\mathcal{D}$. A distribution or a generalized function on
$\mathbb{R}^2$ is a mapping $F : \mathcal{D} \rightarrow
\mathbb{R}$ such that $F$ is $\mathbb{R}$-linear and satisfies the
following  continuity condition. : Let $\phi_k \in \mathcal{D}$
and supp $\phi_k$ be contained in a fixed bounded set for all $k$.
If $\phi_k$ and all their derivatives converge uniformly to zero,
then $F(\phi_k)$ converges to zero.
\end{definition}

  Motivated by the above, from now on, we
 identify locally integrable function $f$ defined on $\mathbb{R}^2$
 with the distribution $\phi \mapsto \int f \phi d\mu$ where $d\mu = dx dy$ is
 the Lebesgue measure. Since we are interested only in continuous functions in
 this paper, and
 continuous functions are locally integrable, we identify
 continuous function $f$ with the distribution $\phi \mapsto \int f \phi
 d\mu$ and vice versa. Therefore, when we say that $f = g$ as
 distributions, it means that $\int f \phi d\mu = \int g \phi
 d\mu$ for all $\phi \in \mathcal{D}$.

We now generalize the notion of derivatives of ordinary functions
to get derivatives of distributions. If $f$ is $C^1$, integration
by parts gives us $\int_{\mathbb{R}^2} \frac{\partial f}{\partial
x}
 \phi \,dx dy =\int_\mathbb{R} \Big( \Big[ f(x,y) \phi(x,y) \Big]_{x=-\infty}^{x=\infty} - \int_\mathbb{R} f \frac{\partial \phi}{\partial x}
\, dx \Big) dy$. Since test function $\phi$ has a compact support,
we have $\phi(-\infty,y) = \phi(\infty,y) =0$ and thus
 $\int_{\mathbb{R}^2} \frac{\partial f}{\partial x}
 \phi \, dx dy = - \int_{\mathbb{R}^2} f \frac{d \phi}{dx}
\, dx dy$. Motivated by this fact, we define the derivatives of
distributions as the following.

 \begin{definition}
 If $F$ is a distribution
on $\mathbb{R}^2$, we define its partial derivatives by
$\frac{\partial F}{\partial x} \big(\phi \big) = -F\big(
\frac{\partial \phi}{\partial x} \big)$ and $\frac{\partial
F}{\partial y} \big(\phi \big) = -F\big( \frac{\partial
\phi}{\partial y} \big)$.
\end{definition}

 From the definition of distribution, it is not difficult to see
 that the derivatives of a distribution are also distributions, and it
  must be noted that any distribution can be differentiated infinitely
many times since test functions are differentiable infinitely many
times. For example, if $f$ is a continuous function defined on
$\mathbb{R}^2$, then its partial derivative $\frac{\partial
f}{\partial x}$ is the distribution $\phi \mapsto -
\int_{\mathbb{R}^2} f(x,y) \frac{\partial \phi}{\partial x} \, dx
dy$. We also remark that for any distribution $F$, we have
$\frac{\partial^2 F}{\partial x \partial y} = \frac{\partial^2
F}{\partial y \partial x}$. To compare the derivative of a
distribution with the derivative in the classical sense, we have
the following.

\begin{thm}
Let $f$ be a real-valued function defined on $\mathbb{R}$, which
is differentiable in the classical sense. Then, its derivative
$f^\prime$ in the classical sense is the same as the derivative in
the distribution sense.
\end{thm}
\begin{proof}
Let $\hat{f}$ denote the derivative of $f$ in the distribution
sense. To show that $f^\prime = \hat{f}$, it is sufficient to show
that they produce the same value when multiplied by test functions
and integrated. By computation, we have $\int_{\mathbb{R}} f^
\prime \phi dx = - \int_{\mathbb{R}} f \phi^\prime dx$ which is
the same as $\hat{f}\big(\phi\big) = - \int_{\mathbb{R}} f
\phi^\prime dx $ by definition.
\end{proof}

 We finally remark that by the same argument, we can show that the above theorem
 also holds for multi-variable functions and for their partial
 derivatives. For example, though $f(x,y) = \sqrt{x^2+y^2}$ does not
have partial derivatives at $(0,0)$ in the classical sense, we can
see that
 its distribution
 $\frac{\partial f}{\partial x}$ is a map $\phi \mapsto -\int
 \sqrt{x^2+y^2} \frac{\partial \phi}{\partial x} \, dxdy$ and it
 is not difficult to see that $\frac{\partial^2 f}{\partial x
 \partial y} = \frac{\partial^2 f}{\partial y \partial x}$ as distributions.
   This is the way in which we generalize derivatives.

\section{Continuous causal isomorphisms} \label{section:3}

From now on, all functions are considered as identified with the
corresponding distributions and derivatives are considered in the
distribution sense.

\begin{proposition} \label{partial}
Let $z=f(x,y)$ where $f$ is locally integrable. Then we have
$\frac{\partial f}{\partial x} = 0$ if and only if $z = f(y)$.
\end{proposition}
\begin{proof}

Assume that $\frac{\partial f}{\partial x} = 0$ and choose a
$C^\infty$ function $\phi_0 : \mathbb{R} \rightarrow \mathbb{R}$
such that $\mbox{supp}\,\, \phi_0$ is compact and
$\int_{-\infty}^{\infty} \phi_0 dx = 1$. For $\phi \in
\mathcal{D}$, let $\psi(x,y) = \phi(x,y) - \phi_0(x)
\int_{-\infty}^{\infty} \phi(u,y) du$. Then, it is easy to see
that $\psi \in \mathcal{D}$. If we let $\eta_\phi(x,y) =
\int_{-\infty}^{x} \phi(u,y) du - \Big(\int_{-\infty}^{\infty}
\phi(u,y) du \Big) \int_{-\infty}^{x} \phi_0(u) du$, then we have
$\frac{\partial \eta_{\phi}}{\partial x}(x,y) = \psi(x,y)$ and
$\eta_{\phi} \in \mathcal{D}$ since $\int_{-\infty}^{\infty}
\phi_0 dx =1$.

Since $\frac{\partial f}{\partial x} = 0$, we have $0 = \int
\frac{\partial f}{\partial x} \eta_{\phi} d\mu = -\int f
\frac{\partial \eta_{\phi}}{\partial x} d\mu = -\int f \psi d\mu$.
If we substitute $\psi(x,y) = \phi(x,y) - \phi_0(x)
\int_{-\infty}^{\infty} \phi(u,y) du$ into $\int f \psi d\mu = 0$,
we have
\begin{eqnarray*}
\int_{\mathbb{R}^2} f \phi \,\, d\mu &=& \int_{\mathbb{R}^2}
f(x,y)
\phi_0(x) \Big( \int_{-\infty}^{\infty} \phi(u,y) du \Big) d\mu \\
&=& \int_{-\infty}^{\infty} \int_{-\infty}^{\infty}
\int_{-\infty}^{\infty} f(x,y) \phi_0(x) \phi(u,y) \,\, du dx dy \\
&=& \int_{-\infty}^{\infty} \int_{-\infty}^{\infty}
\int_{-\infty}^{\infty} f(u,y) \phi_0(u) \phi(x,y) \,\, du dx dy \\
&=& \int_{-\infty}^{\infty} \int_{-\infty}^{\infty} \Big[
\int_{-\infty}^{\infty} f(u,y) \phi_0(u) du \Big] \phi(x,y) \,\,
dx dy.
\end{eqnarray*}
Therefore, we have $f = \int_{-\infty}^{\infty} f(u,y) \phi_0(u)
\,\, du$, which is a function of $y$.

Conversely, we now assume that $z = f(y)$. Then, we have

\begin{eqnarray*}
\frac{\partial f}{\partial x} &=& -f\Big[\frac{\partial
\phi}{\partial
x}\Big]\\
&=& - \int \int f \frac{\partial \phi}{\partial x} \,\, dx dy \\
&=& - \int_{-\infty}^{\infty} f(y) \Big( \int_{-\infty}^{\infty}
\frac{\partial \phi}{\partial x} \,\, dx \Big) dy\\
&=& - \int_{-\infty}^{\infty} f(y) \cdot 0 \,\, dy = 0.
\end{eqnarray*}
In the last equality, we have used the fact that $\phi$ has a
compact support.

\end{proof}

\begin{thm}
Let $f$ and $g$ be locally integrable functions defined on
$\mathbb{R}^2$. If $\frac{\partial f}{\partial x} = g$, then there
exist distributions $h$ and $G$ such that $f = h + G$, where
$\frac{\partial h}{\partial x} = 0$ and $\frac{\partial
G}{\partial x} = g$.
\end{thm}
\begin{proof}

Choose a smooth $\phi_0 : \mathbb{R} \rightarrow \mathbb{R}$ such
that $\int_{-\infty}^{\infty} \phi_0 \,\, dx = 1$ and supp
$\phi_0$ is compact. If we define $\psi$ and $\eta_{\phi}$ in the
same manner as in the Proposition \ref{partial}, then, we have
$\psi \in \mathcal{D}$, $\eta_{\phi} \in \mathcal{D}$ and
$\frac{\partial \eta_{\phi}}{\partial x}(x,y) = \psi(x,y)$.

Since $\frac{\partial f}{\partial x} = g$, we have

$$\int f \frac{\partial \phi}{\partial x} \,\, d\mu = -\int g \phi
\,\, d\mu, \,\,\,\, \mbox{for all} \,\,\,\, \phi \in
\mathcal{D}.$$ In particular, we have
\begin{eqnarray*}
 \int f \frac{\partial \eta_{\phi}}{\partial x} \,\, d\mu &=&
-\int g \eta_{\phi} \,\,
d\mu\\
 \mbox{and so} \,\,\, \int f \psi \,\, d\mu &=& -\int g \eta_{\phi}
\,\, d\mu.
\end{eqnarray*}

If we substitute $\psi(x,y) = \phi(x,y) - \phi_0(x)
\int_{-\infty}^{\infty} \phi(u,y) du$ into the above equation, we
have
$$ \int_{\mathbb{R}^2} f \phi \,\, d\mu - \int_{\mathbb{R}^2}
f(x,y) \phi_0(x) \Big( \int_{-\infty}^{\infty} \phi(u,y) \,\, du
\Big) \,\, d\mu = -\int g \eta_{\phi} \,\, d\mu \,\, \cdots \,\,
(*).
$$

If let $h\big(\phi\big) = \int_{\mathbb{R}^2} \Big(
\int_{-\infty}^{\infty} f(u,y) \phi_0(u) \,\, du \Big) \phi(x,y)
\,\, d\mu$, which is the same as the second term on the left hand
side of $(*)$, then it is easy to see that $h$ is a distribution
and, since $h = \int_{-\infty}^{\infty} f(u,y) \phi_0(u) \,\, du$
as a function, we have $\frac{\partial h}{\partial x} = 0$ by
Proposition \ref{partial}.

If we consider $G : \phi \mapsto - \int g \eta_{\phi} \,\, d\mu$,
then, from $(*)$, $G = f - h$ and so $G$ is a distribution.

To show $\frac{\partial G}{\partial x} = g$, since
$\int_{-\infty}^{\infty} \frac{\partial \phi}{\partial u}(u,y)
\,\, du = 0$, we have

\begin{eqnarray*}
\frac{\partial G}{\partial x}\big( \phi \big) &=& -G\big(
\frac{\partial \phi}{\partial
x} \big)\\
&=& \int g \Big[ \int_{-\infty}^{x} \frac{\partial \phi}{\partial
u}(u,y) \,\, du - \Big( \int_{-\infty}^{\infty} \frac{\partial
\phi}{\partial u}(u,y) \,\, du \Big) \int_{-\infty}^{x} \phi_0(u)
\,\, du \Big] \,\, d\mu \\
&=& \int g(x,y) \phi(x,y) \,\, d\mu, \,\, \mbox{since} \,\, \phi
\,\, \mbox{has compact support.}
\end{eqnarray*}

Therefore, we have $\frac{\partial G}{\partial x} = g$.

\end{proof}

We now prove one of the key theorems which characterize continuous
causal isomorphisms.

\begin{thm} \label{main}
Let $f=f(x,y)$ be a locally integrable function. Then,
$\frac{\partial^2 f}{\partial x \partial y} = 0$ if and only if
$f(x,y) = \alpha(x) + \beta(y)$, where $\alpha$ and $\beta$ are
locally integrable functions. Furthermore, if $f$ is continuous,
then $\alpha$ and $\beta$ are continuous.
\end{thm}

\begin{proof}

Assume that $\frac{\partial^2 f}{\partial x \partial y} = 0$ and
choose a $C^\infty$ function $\phi_0(x)$ such that
$\int_{-\infty}^{\infty} \phi_0 dx =1$ and the support of $\phi_0$
is compact. Let $\phi_1(x,y) = -\phi_0(x)\phi_0(y)$. Then,
$\phi_1$ is $C^\infty$, has compact support, and we have
$\int_{-\infty}^\infty \phi_1(u,y) du = -\phi_0(y)$ and
$\int_{-\infty}^\infty \phi_1(x,u) du = -\phi_0(x)$. Given any
test function $\phi$, define $\psi$ and $\eta_\phi$ by
$$ \psi(x,y) = \phi(x,y) - \phi_0(x)\int_{-\infty}^{\infty}
\phi(u,y) du - \phi_0(y) \int_{-\infty}^{\infty} \phi(x,v) dv -
\phi_1(x,y) \Big(\int \phi d\mu\Big)$$ and

\begin{eqnarray*}
 \eta_\phi(x,y) &=& \int_{-\infty}^x \int_{-\infty}^y \phi(u,v) dv
du - \int_{-\infty}^x \phi_0(u) du \int_{-\infty}^{\infty}
\int_{-\infty}^y \phi(u,v) dv du\\
& & - \int_{-\infty}^y \phi_0(u) du
\int_{-\infty}^{\infty}\int_{-\infty}^x \phi(u,v) du dv -
\int_{-\infty}^x \int_{-\infty}^y \phi_1(u,v) du dv \Big(\int
\phi d\mu \Big)
\end{eqnarray*}

 In the definition of $\eta_\phi$, we can see that, for large $x$,
 the first and second terms cancel out, and the third and fourth
 terms cancel out. Also, for large $y$, the first and third terms
 cancel out and the second and fourth terms cancel out. Therefore,
 $\eta_\phi$ has a compact support. Now it is easy to see that $\psi$ and $\eta_\phi$ are
 test functions and $\frac{\partial^2 \eta_\phi}{\partial x
 \partial y} = \psi(x,y)$. Since $f_{xy} = \frac{\partial^2
 f}{\partial x \partial y} = 0$, we have $0 = \int \int f_{xy}
 \phi(x,y) dx dy = \int \int f(x,y) \phi_{xy} dx dy$ for any test
 functions $\phi(x,y)$. Therefore, we have $0 = \int \int f(x,y)
 \frac{\partial^2 \eta_\phi}{\partial x \partial y} dx dy = \int
 \int f(x,y) \psi(x,y) dx dy$.

 If we substitute the definition of $\psi$ into $\int f \psi d\mu = 0$, we have

\begin{eqnarray*}
\int\int f \phi dx dy &=& \int \int f(x,y) \phi_0(x) \Big(
\int_{-\infty}^{\infty} \phi(u,y) du \Big) dx dy\\
& & + \int \int f(x,y) \phi_0(y) \Big(\int_{-\infty}^{\infty}
\phi(x,v) dv \Big)
dx dy\\
& & + \int \int f(x,y) \phi_1(x,y) dx dy \Big( \int \phi d\mu \Big) \\
&=&
\int_{-\infty}^{\infty}\int_{-\infty}^{\infty}\int_{-\infty}^{\infty}
f(x,y) \phi_0(x) \phi(u,y) du dx dy \\
& & +
\int_{-\infty}^{\infty}\int_{-\infty}^{\infty}\int_{-\infty}^{\infty}
f(x,y)\phi_0(y)\phi(x,v) dv dx dy \\
& & + \Big( \int_{-\infty}^{\infty}\int_{-\infty}^{\infty} f(x,y)
\phi_1(x,y) dx dy \Big)
\Big(\int_{-\infty}^{\infty}\int_{-\infty}^{\infty} \phi dx dy
\Big)\\
&=& \int_{-\infty}^{\infty}\int_{-\infty}^{\infty}
\Big[\int_{-\infty}^{\infty}f(u,y) \phi_0(u) du \Big] \phi(x,y) dx
dy \\
& & + \int_{-\infty}^{\infty}\int_{-\infty}^{\infty}
\Big[\int_{-\infty}^{\infty} f(x,v) \phi_0(v) dv \Big] \phi(x,y)
dx dy \\
& & + \int_{-\infty}^{\infty}\int_{-\infty}^{\infty} c \phi(x,y)
dx dy,\\
 \,\,\, \mbox{where}\,\,\, c &=&
\int_{-\infty}^{\infty}\int_{-\infty}^{\infty} f(x,y) \phi_1(x,y)
dx dy.
\end{eqnarray*}

In other words, $f(x,y) = \alpha(x) + \beta(y)$ where $\alpha(x) =
\int_{-\infty}^{\infty}f(x,v) \phi_0(v) dv$ and $\beta(y) =
\int_{-\infty}^{\infty} f(u,y) \phi_0(u) du + c$. Since $f$ is
locally integrable, from the definition of $\alpha$ and $\beta$ we
can see that $\alpha$ and $\beta$ are locally integrable. Since
any continuous map defined on compact set is uniformly continuous,
if $f$ is continuous, then $\alpha$ and $\beta$ are continuous.

 Since $f_{xy} = f_{yx}$ for any distribution $f$, the converse is
 easily obtained from Proposition \ref{partial}.

\end{proof}

\section{The characterization} \label{section:4}

 On two-dimensional Minkowski spacetime $\mathbb{R}^2_1$, if we
 use null coordinates $u = x+t$ and $v = x-t$, then we have the
 following.

\begin{thm}
Let $(U,V) = F(u,v)$ be a causal isomorphism on $\mathbb{R}^2_1$.
Then, there exist unique homeomorphisms $\varphi$ and $\psi$ on
$\mathbb{R}$, which are both increasing or both decreasing such
that if $\phi$ and $\psi$ are increasing, then we have $F(u,v) =
\big( \varphi(u), \psi(v) \big)$, or if $\varphi$ and $\psi$ are
decreasing, then we have $F(u,v) = \big( \varphi(v), \psi(u)
\big)$.

Conversely, for any given homeomorphisms $\varphi$ and $\psi$ on
$\mathbb{R}$, which are both increasing or both decreasing, the
function $F$ defined as above is a causal isomorphism on
$\mathbb{R}^2_1$.
\end{thm}
\begin{proof}
See Theorem 2.2 in \cite{Kim}.
\end{proof}

 We now begin to characterize continuous causal isomorphisms by
 use of results obtained in the previous section.

\begin{thm}
Let $(\sigma, \tau) = F(u,v) $ be a homeomorphism from
$\mathbb{R}^2_1$ onto $\mathbb{R}^2_1$ where $(u,v)$ and $(\sigma,
\tau)$ are null coordinates. Suppose that, for any function
$\theta$ on $\mathbb{R}^2_1$, $\theta_{uv}=0$ if and only if
$\theta_{\sigma \tau}=0$. Then, there exist homeomorphisms
$\varphi$ and $\psi$ on $\mathbb{R}$ such that, either $F(u,v) =
\big( \varphi(u), \psi(v) \big)$ or $F(u,v) = \big( \varphi(v),
\psi(u) \big)$.
\end{thm}

\begin{proof}
If we let $\theta(\sigma, \tau)=\sigma$, then $\theta$ satisfies
$\theta_{\sigma \tau}=0$ and so, by assumption, we have
$\theta_{uv}=0$. Then, by Theorem \ref{main}, there exist
continuous functions $\alpha$ and $\beta$ such that $\sigma(u,v) =
\alpha(u) + \beta(v)$. Likewise, by considering $\theta(\sigma,
\tau) = \tau$, we can get continuous functions $\gamma$ and
$\delta$ such that $\tau(u,v) = \gamma(u) + \delta(v)$. From
 Proposition \ref{partial}, we can see that $\theta(\sigma, \tau) =
\sigma^2$ satisfies $\theta_{\sigma \tau} = 0$ and thus, we have
$\frac{\partial \sigma^2}{\partial u \partial v} = 0$. Since
$\sigma(u,v)^2 = \alpha(u)^2 + 2\alpha(u)\beta(v) + \beta(v)^2$,
and $\frac{\partial^2}{\partial u \partial v} \Big( \alpha(u)^2 +
\beta(v)^2 \Big) = 0$ by Theorem \ref{main}, we have
$\frac{\partial^2}{\partial u
\partial v} \Big( \alpha(u) \beta(v) \Big) = 0$. By Theorem
\ref{main} again, we have $\alpha(u) \beta(v) = f(u) + g(v)$,
where $f$ and $g$ are continuous. If we assume that $\alpha$ is
not constant, there exist $u_1$ and $u_2$ such that $\alpha(u_1)
\neq \alpha(u_2)$ and so we have
\begin{eqnarray*}
\alpha(u_1) \beta(v) &=& f(u_1) + g(v) \,\,\, \mbox{and} \\
 \alpha(u_2) \beta(v) &=& f(u_2) + g(v).
\end{eqnarray*}

Therefore, we have $\beta(v) = \frac{f(u_1)-f(u_2)}{\alpha(u_1) -
\alpha(u_2)}$. In other words, if $\alpha$ is not constant, then
$\beta$ is a constant function. Since $(\sigma, \tau) = F(u,v)$ is
a bijection, if $\alpha$ is a constant function, then $\beta$ is
not a constant function.

By considering $\theta(\sigma, \tau) = \tau^2$, by the same
argument as above, we can show that $\gamma$ is not a constant
function if and only if $\delta$ is a constant function.

Assume first that $\alpha$ is not a constant function. Then, by
the previous observation, we have $\sigma = \varphi(u)$ where
$\varphi(u) = \alpha(u) + \, \mbox{constant}$ is a continuous
function. If $\gamma(u)$ is not a constant function, then $\tau =
\gamma(u) + c$ and thus $F$ is not a bijection, which is a
contradiction. Therefore, $\gamma$ is a constant function and thus
we have $\tau = \psi(v)$ where $\psi(v) = \delta(v) + \,
\mbox{constant}$ is a continuous function. Since $F$ is bijective,
$\varphi$ and $\psi$ must be bijective and by the topological
domain of invariance, continuous bijections $\varphi$ and $\psi$
are homeomorphisms on $\mathbb{R}$.

We now assume that $\alpha$ is a constant function. Then, since
$\beta$ is not a constant function, we have $\sigma = \varphi(v)$
where $\varphi(v) = \beta(v) + \, \mbox{constant}$ is a continuous
function and, by the same argument as the above, we can show that
$\tau = \psi(u)$ where $\psi(u) = \gamma(u) + \, \mbox{constant}$
is a continuous function. By the same argument as the above,
$\varphi$ and $\psi$ are homeomorphisms on $\mathbb{R}$.
\end{proof}

By combining the above two theorems, we have the following.

\begin{thm}
Let $F : \mathbb{R}^2_1 \rightarrow \mathbb{R}^2_1$ be a
homeomorphism given by $(\sigma, \tau) = F(u,v)$ where $(u,v)$ and
$(\sigma, \tau)$ are null coordinates. Then, a necessary and
sufficient condition for $F$ to be a causal isomorphism is that,
for any continuous function $\theta$ on $\mathbb{R}^2_1$, we have
$\theta_{uv} = 0$ if and only if $\theta_{\sigma \tau} = 0$, and
either (1) $\sigma$ and $\tau$ are increasing functions of $u$ and
$v$, respectively, or (2) $\sigma$ and $\tau$ are decreasing
functions of $v$ and $u$, respectively.
\end{thm}

Since $\mathbb{R}^2_1$ is strongly causal, any causal isomorphism
on $\mathbb{R}^2_1$ is a homeomorphism and thus, this theorem
characterizes all of the causal isomorphisms on $\mathbb{R}^2_1$
by invariance of wave equations, regardless of their smoothness.
With this theorem and the results in \cite{Wave} and \cite{Kim},
 the characterization of causal isomorphisms in terms of wave equations is
 completed regardless of their smoothness and their dimensions.

\section{Acknowledgement}

The present research was conducted by the research fund of Dankook
University in 2014.

\end{document}